\documentclass[12pt,reqno]{amsart}

\usepackage{mystyle}
\usepackage{enumitem}

\newcommand{\pot}{\mathrm{pot}}

\newcommand{\Real}{\mathbb{R}}
\newcommand{\Torus}{\mathbb{T}}

\newcommand{\bbZ}{\mathbb{Z}}

\DeclareMathOperator{\Coulomb}{Coulomb}

\title{A simple bound on fluctuations in the 3D Coulomb gas}

\author{Alex Cohen}
\address{Courant Institute, New York University. New York, NY, USA.}
\email{alexcohen@nyu.edu}

\author{Felipe Hern\'andez}
\address{Department of Mathematics, Penn State University. State College, PA, USA.}
\email{felipeh@psu.edu}

\begin{document}
\begin{abstract}
The Coulomb gas models an interacting system of $N$ negatively charged particles. 
We give a new proof that, at sufficiently low temperature, smooth linear statistics $\sum_j \varphi(x_j)$ are bounded by $C N^{1-2/d}$. 
\end{abstract}
\maketitle 

\section{Introduction}
The one-component plasma, also known as the Coulomb gas or jellium, describes a system of $N$ particles $X = (x_1,\dots,x_n)\in (\Real^d)^N$ of negative charge
in background of uniform positive charge density.  This is the simplest model of interacting charged particles, and it appears to describe certain physical systems such as charged colloidal mixtures~\cite{colloid}.  Of particular interest in this model is the distribution of particles at equilibrium, which is \textit{hyperuniform}~\cites{JLM,Leb}: a simple physical argument based on Gauss's law~\cite{Martin} predicts that the fluctuation of the number of particles $N_\Omega$ in a macroscopic 
region $\Omega\subset\Real^3$ should scale not with the volume $|\Omega|$ but instead with the surface area $|\partial \Omega|$.
This prediction remains conjectural.

Strong results have been established in dimension two~\cite{BBNY,Leble,Thoma,NishryYakir}, and in dimension
one the Coulomb gas is the familiar $\beta$-ensemble from random matrix theory.

In this paper we consider the fluctuations of \textit{smooth} linear statistics of the form $\langle \mu_X, \varphi\rangle := \sum \varphi(x_j)$ for some test function $\varphi\in C^2(\Real^n)$.  
If the particles were placed independently at random, the fluctuations would have order $\sim N^{1/2}$. 
Are the fluctuations of the Coulomb gas smaller than that?
The first progress in $d = 3$ was made by Chatterjee~\cite{Chatterjee}, and further developed by Ganguly-Sarkar~\cite{GangulySarkar}, who proved hyperuniformity for a related hierarchical model.
In $\R^d$ for $d \geq 3$, Serfaty~\cite{Serfaty2020} proved a number of results about the Coulomb gas at different length scales and temperatures, including that (over a large range of temperatures) the order of fluctuations is bounded by $C N^{1 - 2/d}$. In $d = 3$, this is less than that of independent particles. 
She also proved a sharper conditional result describing a central limit theorem for fluctuations. 

In this note, we give a new proof of the fluctuation bound $C N^{1-2/d}$. 
Our proof is based on the idea that one can prove fluctuations for linear statistics by first proving a bound on the fluctuations of the potential function $P\mu_X = \Delta^{-1} \mu_X$, via the identity
\[
|\langle \varphi,\mu_X\rangle| = |\langle \Delta \varphi, P\mu_X\rangle| \leq \|\Delta \varphi\|_{L^\infty}\|P\mu_X\|_{L^1}.
\]
This identity works on the torus, and a variant taking into account a confining potential can also be used on $\Real^n$.
The point then is that, because the potential has mean zero, the $L^1$ norm $\|P\mu_X\|_{L^1}$ 
can only be very large if there are some regions with 
extremely large negative potential.  Intuitively, this can be controlled by observing the Glauber dynamics. If $P\mu_X$ is very negative somewhere, then a resampled particle would prefer to move to this region and re-balance the potential somewhat.  

On $\Real^d$ one must adapt this argument to incorporate the effect of the confining potential, and 
we only obtain results for observables that are contained in the support of the limiting equilibrium measure.  In this case, our result is not as self-contained, relying for example on estimates for the minimum energy configuration due to Serfaty~\cite{Serfaty}.

We note that our method applies to macroscopic observables, whereas the transport method in~\cite{Serfaty2020} gives fluctuation bounds for observables supported on small length scales. 

\subsection{Result on the torus $\Real^d / \bbZ^d$}
We define a Gibbs measure on $(\Torus^d)^N := (\Real^d/\bbZ^d)^N$ using the Hamiltonian
\[
\mc{H}(X) := \sum_{1\leq i<j\leq N} g_{\rm Coulomb}(x_i-x_j), 
\]
where $g_{\rm Coulomb}$ solves
\[
\Delta g_{\rm Coulomb} = -\delta_0 + 1.
\]
We then define the Gibbs measure $d\P_{N,\beta}$ on $(\Torus^d)^N$ by\footnote{Note that this differs from the definition of~\cite{Serfaty} by a multiplicative factor of $N^{\frac{d-2}{d}}$ in $\beta$.  This is just a scaling convention.}
\[
d\P_{N,\beta}(X) = \frac{1}{Z_{N,\beta}}e^{-\beta \mc H_N(X)}dX,
\]
with partition function 
\[
Z_{N,\beta} := \int_{(\Torus^d)^N} e^{-\beta \mc H_N(X)} \,dX.
\]
The empirical measure of a configuration $X$ is given by
\[
\mu_X := \sum_{j=1}^N \delta_{x_j}.
\]
Our main result is the following concentration inequality for smooth statistics $\langle \varphi, \mu_X\rangle$.
\begin{theorem}[Bound for fluctuations on $\Torus^d$]
\label{thm:mainTorus}
For $d\geq 3$ there exists some absolute constant $C_d$ such that for $A\geq C_d$ the following bound holds for $X$ sampled
from the Gibbs measure $\P_{N,\beta}$ and $\phi\in C^2(\Torus^d)$:
\begin{equation}
\label{eq:fluct-bd}
P_{N,\beta}( |\langle \phi, \mu_X\rangle -1_{\Torus^d}\rangle| \geq A N^{1-2/d}\|\Delta \phi\|_{L^\infty})  \leq 2 e^{-\frac{1}{2} A \beta N^{1-2/d}}. 
\end{equation}
\end{theorem}
Theorem~\ref{thm:mainTorus} shows that for smooth $\varphi$, fluctuations of $\langle \phi,\mu_X\rangle$ are on the order $N^{1/3}$ in dimension $d=3$ and for any $\beta \geq N^{-1/3}$ (corresponding to temperature $T\leq N^{1/3}$).
In this regime of macroscopic observables in $d=3$ and $T\simeq N^{1/3}$, our result matches the fluctuation bound of Serfaty~\cite{Serfaty2020}, also see~\cite[Corollary 1.4]{Serfaty}.  

\subsection{Result on $\Real^d$}
On Euclidean space we need to consider a Hamiltonian with a confining potential $V:\R^d\to\R$.  
In this case, for $X\in (\R^d)^N$ the Hamiltonian reads 
\begin{align*}
    \mc H_N(X) := \frac{1}{2}\sum_{j\neq k}g(x_j - x_k) + N\sum_{1\leq j \leq N} V(x_j),
\end{align*}
where we take $\Delta g = -\delta_0$ to be the Coulomb interaction.
Again we consider the Gibbs measure 
\[
d\P_{N,\beta}(X) = \frac{1}{Z_{N,\beta}} e^{-\beta \mc H_N(X)}\, dX
\]
where the partition function 
\[
Z_{N,\beta} = \int_{(\R^d)^N} e^{-\beta \mc H_N(X)}\, dX
\]
is bounded as long as $\int \exp(-V(x))\, dx < \infty$. 

In the limit $N\to \infty$ the one-particle distribution of $X$ converges to an equilibrium measure which is the minimizer
of the following functional
\begin{align*}
    \mc E(\mu) := \frac{1}{2}\int \int g(x-y)\, d\mu(x) d\mu(y) + \langle 1, \mu\rangle \int V(x)\, d\mu(x).
\end{align*}

A theorem due to Frostman and recalled below in Theorem~\ref{thm:FrostmanThm} guarantees the existence of compactly supported minimizers
$\mu_V$ for $\mcal{E}$.  Note that in the special case $d = 3$, and $V(x) = \frac{1}{2}|x|^2$ is the quadratic potential, then $\mu_V$ is the Lebesgue probability measure on a ball centered at the origin. 

Following Serfaty~\cite{Serfaty}, we assume $V$ satisfies the following properties. The assumptions simplify compared to \cite{Serfaty} because we assumed $g \geq 0$. 
\begin{enumerate}[label=(B\arabic*)]
    \item $V$ is lower semi-continuous and bounded below, \label{ass:V_1}

    \item $\lim_{|x|\to \infty} V(x) = \infty$, \label{ass:V_2}

    \item $\{x\in \R^d\, :\, V(x) < +\infty\}$ has positive capacity, \label{ass:V_3} 

    \item $\int_{\R^d} \exp(-V(x))\, dx < \infty$, \label{ass:V_4}

    \item The equilibrium measure is bounded in $L^{\infty}$, that is, $\mu_V = \mu_V(x) dx$ for some $\mu_V(x) \in L^{\infty}(\R^d)$. \label{ass:V_5}
\end{enumerate}
In particular, the potential $V(x) = \frac12|x|^2$ is valid.
With these assumptions we have the following result about the fluctuations of smooth observables.
\begin{theorem}\label{thm:fluctuations_Rd}
Let $\mc H_N$ be the Coulomb interaction in $d\geq 3$ with external potential $V$ satisfying \ref{ass:V_1}---\ref{ass:V_5}, and let $N\geq 2$, $\beta \geq 1/(N-1)$. Assume the equilibrium measure $\mu_V$ lies in $L^{\infty}$. Then there exist constants $c$ and $C$ depending on $V$ such that, for any smooth observable $\varphi \in C_c^2(\Sigma)$ supported in $\Sigma$ with $\langle \varphi,\mu_V\rangle=0$,
\begin{align*}
    \P_{N,\beta}(|\langle \varphi, \mu_X\rangle| \geq AN^{1-2/d} \| \Delta \phi \|_{L^{\infty}}) \leq Ce^{\beta N^{1-2/d}(C - cA)}. 
\end{align*}
\end{theorem}
In particular, this result implies that in $d=3$ the fluctuations of $C^2$ observables are on the order $N^{1/3}$.  

\subsection{Organization of the Paper}
The result on the torus, Theorem~\ref{thm:mainTorus} is proven in Section~\ref{sec:torus}.  The generalization of the argument to the full
Euclidean space and the derivation of Theorem~\ref{thm:fluctuations_Rd} is then completed in Section~\ref{sec:Rd}.

\subsection{Acknowledgements}
Thanks to Sylvia Serfaty for helpful comments on this paper.  FH was supported by NSF grant DMS-2303094. Alex Cohen was supported by a Clay Research Fellowship. 

\section{Coulomb gas on the torus}
\label{sec:torus}
\subsection{The probabilistic estimate}
Our result on the torus applies to a slightly more general setup than just the Coulomb interaction, and it is clarifying to precisely what properties of the interaction $g(x,y)$ are used.  In particular, we consider interaction potentials $g$ that are:
\begin{enumerate}[label=(A\arabic*)]
    \item Symmetric: for all $x,y\in\Torus^d$, $g(x,y) = g(y,x)$, \label{assT:1}
    \item Bounded in $L^1$: for all $x\in\Torus^d$, $\int_{\Torus^d} |g(x,y)|dy < \infty$ \label{assT:2}
    \item Mean zero: for all $x\in\Torus^d$, $\int g(x,y)\, dy = 0$. \label{assT:3}
    \item Bounded from below: there exists $m_{\pot} < \infty$ such that $g(x,y) \geq -m_{\pot}$. \label{assT:4}
\end{enumerate}
Our model case is the Coulomb potential $g_{\Coulomb}(x,y) = g_{\Coulomb}(x-y)$, which is given by
\begin{equation}
\label{eq:g-torus-def}
    g_{\Coulomb}(x) = \sum_{\xi \in \Z^d\setminus \{0\}} \frac{1}{4\pi |\xi|^2} e^{2\pi i \xi \cdot x}.
\end{equation}
Letting $\Delta = \partial_{x_1}^2 + \dots + \partial_{x_d}^2$ be the Laplacian, the Coulomb potential solves the equation
\begin{align}\label{eq:LaplacianEqCoulombTorus}
\Delta g_{\Coulomb} = -\delta_0 + 1_{\T^d}.
\end{align}
The Coulomb potential is smooth away from $x = 0$, and its asymptotic behavior near zero is the same as in $\R^d$. In particular, for $d \geq 3$
\begin{align}
\label{eq:GreenFuncAsymptotic_Torus}
g_{\Coulomb}(x) = c |x|^{2-d} + \text{(Smooth function)}.
\end{align}
Once this asymptotic and the Fourier expansion are known, \ref{assT:1}--\ref{assT:4} are immediate. 

The quantity we consider is the potential field $P\mu_X$ of the empirical measure, where
\begin{equation}\label{eq:potential_field}
    P\mu(x) = \int g(x,y)d\mu(y).
\end{equation}
The potential $P\mu_X$ determines where a putative additional particle would be placed.  In other words, if $P\mu_X$ is very negative in some region, then it might be energetically favorable to reposition a single particle to that region.
Our main contribution is a probabilistic estimate for the $L^1$ norm of the potential field $P\mu_X$ in terms of the 
energy of the ground state configuration, $\min_X \mc H_N(X)$.   
\begin{lemma}
\label{lem:Lone-moment-torus}
Let $g(x,y)$ be a potential satisfying \ref{assT:1}--\ref{assT:4}, $\mc H_N$ the corresponding Hamiltonian, and $\P_{N,\beta}$ the associated Gibbs measure. We have the exponential moment bound
\[
\E_{X\sim \P_{N,\beta}}[ e^{\frac12 \beta \| P\mu_X\|_{L^1}}]
\leq e^{\beta m_{\rm pot}} e^{-\frac{2\beta}{N} \min_X \mc{H}_N(X)} + 1.
\]
In particular, this estimate holds for the Coulomb interaction.
\end{lemma}
We can use this estimate for the potential field generated by $X$ to estimate the fluctuation of linear observables 
\[
\langle \phi,\mu_X\rangle = \sum_{1 \leq j \leq N} \phi(x_j),
\]
where $\phi: \T^d \to \C$ is a $C^2$ function. 
Without loss of generality, we may assume that $\phi$ has mean zero (as the integral
of $\mu_X$ does not fluctuate). By \eqref{eq:LaplacianEqCoulombTorus}, $\Delta P\mu_X = N 1_{\T^d} - \mu_X$, and using self-adjointness to move the Laplacian onto $\Delta$, we obtain
\[
|\langle \phi,\mu_X\rangle| = |\langle \Delta \phi, P\mu_X \rangle|
\leq \|\Delta\phi\|_{L^\infty} \|P\mu_X\|_{L^1}.
\]
Taking an exponential moment, \Cref{lem:Lone-moment-torus} implies 
\[
\E_{X\sim \P_{N,\beta}}\bigl[ \exp\bigl(\frac12 \beta \frac{|\langle \phi,\mu_X\rangle|}{\|\Delta\phi\|_{L^\infty}}\bigr) \bigr]
\leq e^{\beta m_{\rm pot}} e^{-\frac{2\beta}{N} \min_X \mc{H}_N(X)} + 1.
\]
Thus
\begin{align}\label{eq:GeneralFlucEstTorus}
\P_{N,\beta}\bigl[|\langle \phi,\mu_X\rangle| \geq \lambda\|\Delta\phi\|_{L^\infty} \bigr] \leq e^{\beta (m_{\pot} - \frac{2}{N} \min_X \mc H_N(X) - \frac{1}{2}\lambda)} + e^{-\frac{1}{2}\beta \lambda}.
\end{align}
In order to use this estimate, we need to lower bound the ground state energy. The following lemma provides this bound for the Coulomb interaction. 
\begin{lemma}
    \label{lem:min-energy-torus}
    For the Coulomb interaction $g_{\Coulomb}$ defined in~\eqref{eq:g-torus-def} in $d\geq 3$, we have
    \[
    \min_X \mc{H}_N(X) \geq -C_{\pot} N^{2-\frac2d}.
    \]
\end{lemma}
Plugging this Lemma into \eqref{eq:GeneralFlucEstTorus} yields 
\[
\P_{N,\beta}\bigl[|\langle \phi,\mu_X\rangle| \geq \lambda\|\Delta\phi\|_{L^\infty} \bigr] \leq e^{\beta (m_{\pot} + 2C_{\pot} N^{1-\frac{2}{d}} - \frac{1}{2}\lambda)} + e^{-\frac{1}{2}\beta \lambda}.
\]
By taking $\lambda = A N^{1-2/d}$ with $A$ large enough, we obtain Theorem~\ref{thm:mainTorus}.


\subsection{Proof of Lemma~\ref{lem:Lone-moment-torus}}
A key idea in the proof is to look at the contributions to the total energy
$\mc{H}_N$ made by each individual particle.  The contribution from each particle cannot be too large. If one particle experiences a lot of energy from nearby particles, it would prefer to move to a location where the potential field is smaller. Intuitively, one can imagine a Glauber dynamics where each particle is individually resampled while the others remain fixed. In our proof we will use exponential moments instead of Glauber dynamics. 

We define for each $j \in \{1,\ldots,N\}$ the configuration $X_{\hat j}$ of all particles excluding $x_j$, 
\begin{align*}
    X_{\hat j} &= (x_1, \ldots, x_{j-1}, x_{j+1}, \ldots, x_N) \\ 
    \mu_{X,\hat j} &= \sum_{k\neq j} \delta_{x_j}.
\end{align*}
The local energy experienced by the point $x_j$ is $P\mu_{X,\hat{j}}(x_j)$, that is
\[
P \mu_{X,\hat{j}} (x_j) = \sum_{k\not=j} g(x_k,x_j).
\]
The total energy is equal to the local energy at $x_j$ plus the interaction energy of the remaining $N-1$ particles, which does not depend on $x_j$:
\begin{align*}
    \mc{H}_N(X) = P\mu_{X,\hat j}(x_j) + \mc{H}_{N-1}(X_{\hat j}).
\end{align*}
The total energy can also be written as a sum of local energies,
\begin{equation}\label{eq:sum_local_energies}
    \mc H_N(X) = \frac{1}{2}\sum_{1\leq j \leq n} P\mu_{X,\hat j}(x_j)\qquad \text{for $1\leq j \leq N$.}
\end{equation}
 
\noindent The conditional probability of sampling $X = (x_1, \ldots, x_{j-1}, x, x_{j+1}, \ldots, x_N)$ after fixing all but the $j$th coordinate is a Gibbs measure in terms of the local energy at $x$,
\begin{align*}
    d\P_{N,\beta\,|\,X_{\hat j}}(x) = \frac{e^{-\beta P\mu_{X,\hat j}(x)}dx}{\int e^{-\beta P\mu_{X,\hat j}(y)}\, dy}. 
\end{align*}
Using this conditional measure, we compute the exponential moment 
\begin{align*}
    \E_{x\sim \P_{N,\beta\,|\,X_{\hat j}}}[e^{\beta P\mu_{X, \hat j}(x)}] &= \frac{1}{\int e^{-\beta P\mu_{X, \hat j}(x)}\, dx}. 
\end{align*}


\noindent Let $\hat{\P}_{N,\beta}$ be the $(N-1)$-particle marginal distribution for
$\hat{X} = (x_1,\cdots,x_{N-1})$ (which is the same as the distribution of $X_{\hat{j}}$).  That is, $\hat{\P}_{N,\beta}$ is the measure satisfying
\[
\int f(\hat{X})\,d\hat{\P}_{N,\beta}(\hat{X})
:= \int f(\hat{X}) \,d\P_{N,\beta}(X).
\]
Integrating the exponential moment bound against this measure we have
\begin{align*}
    \int \Bigl(\int e^{-\beta P\mu_{X, \hat j}(x)}\,dx\Bigl) \E_{x\sim \P_{x,N,\beta|X_{\hat j}}}\Bigl[e^{\beta P\mu_{X,\hat{j}}(x)}\Bigl] d\hat{\P}_{N,\beta}(X_{\hat j}) = 1,
\end{align*}
which by the definition of the marginal measure is equivalent to the identity
\[
\E_{X\sim \P_{N,\beta}} \Bigl
[\Bigl(\int e^{-\beta P\mu_{X, \hat j}(x)}\,dx\Bigl) e^{\beta P\mu_{X,\hat{j}}(x)}
\Bigl] = 1.
\]
It is useful to replace the restricted potential field
$P \mu_{X,\hat j}$ by the full potential field $P\mu_X$.  To
do so we use that the potential is bounded from below by $-m_{\pot}$,
\[
\int e^{-\beta P\mu_X(x)}\,dx
= \int e^{-\beta g(x_j,x)} e^{-\beta P\mu_{X,\hat j}(x)} \,dx
\leq e^{\beta m_{\rm pot}} \int e^{-\beta P\mu_{X,\hat j}(x)}\,dx.
\]
Thus
\[
\E_{X\sim \P_{N,\beta}} \Bigl
[\Bigl(\int e^{-\beta P\mu_{X}(x)}\,dx\Bigl) e^{\beta P\mu_{X,\hat{j}}(x)}
\Bigl]\leq e^{\beta m_{\pot}}.
\]
The left hand side does not depend on the choice of $j$. Summing over $1\leq j \leq N$ and using linearity of expectation as well as Jensen's inequality, 
\begin{align*}
    e^{\beta m_{\rm pot}} &\geq \E_{X\sim \P_{N,\beta}} \Bigl[\Bigl(\int e^{-\beta P\mu_X(x)}\, dx\Bigr)\frac{1}{N}\sum_{1\leq j \leq N} e^{\beta P\mu_{X,\hat j}(x)} \Bigr] \\ 
    &\geq \E_{X\sim \P_{N,\beta}} \Bigl[\Bigl(\int e^{-\beta P\mu_X(x)}\, dx\Bigr)\exp(\frac{1}{N}\sum_{1\leq j \leq N}\beta P\mu_{X,\hat j}(x)) \Bigr] && \text{(Jensen's inequality),}\\ 
    &= \E_{X\sim \P_{N,\beta}} \Bigl[\Bigl(\int e^{-\beta P\mu_X(x)}\, dx\Bigr)\exp(\frac{2\beta}{N} \mc H_N(X)) \Bigr] && \text{(Eq. \eqref{eq:sum_local_energies})}.
\end{align*}
Rearranging, we obtain
\begin{equation}\label{eq:probabilistic_bd_torus}
    \E_{X\sim \P_{N,\beta}} \Bigl[\int e^{-\beta P\mu_X(x)}\, dx\Bigr] \leq e^{\beta m_{\pot}}e^{-\frac{2\beta}{N} \min_X \mc H_N(X)}.
\end{equation}
To obtain the desired $L^1$ bound for $P\mu_X$, we look at the negative part of $P\mu_X$,
$(P\mu_X)_{-}(x) = \min\{P\mu_X(x), 0\}$.  Since $g$ is mean zero it follows also that $\int P\mu_X(x)\, dx = 0$.  In particular,
\begin{align*}
    \int (P\mu_X)_{-}(x)\, dx = -\frac{1}{2} \| P\mu_X \|_1. 
\end{align*}
We estimate the integrand of \eqref{eq:probabilistic_bd_torus} using  
$e^{-\beta P\mu_X(x)}\geq e^{-\beta (P\mu_X)_-(x)} - 1$
\begin{align*}
    \int e^{-\beta P\mu_X(x)}\, dx &\geq \int e^{-\beta (P\mu_X)_{-}(x)}\, dx - 1 \\ 
    &\geq \exp(-\beta \int (P\mu_X)_- dx) - 1 && \text{(Jensen's inequality)}\\ 
    &= e^{\frac{\beta}{2} \| P\mu_X \|_1} - 1.
\end{align*}
Thus \eqref{eq:probabilistic_bd_torus} implies that 
\begin{equation}\label{eq:exponential_moment_L1}
    \E[e^{\frac{\beta}{2} \| P\mu_X \|_1}] \leq e^{\beta(-\frac{2}{N} \min_X \mc H_N(X)+m_{\pot})}+1,
\end{equation}
as desired.

\subsection{The minimum energy configuration}
In this section we prove Lemma~\ref{lem:min-energy-torus} lower bounding the ground state energy for the Coulomb interaction.  For each $0<r<1$ let $\gamma_r = |B_r|^{-1} {\bf 1}_{|x|\leq r}$ be the
$L^1$-normalized indicator function of the ball of radius $r$. 

We use the following bounds for the Coulomb potential:
\begin{align}
    \gamma_r \ast g_{\Coulomb} &\leq g_{\Coulomb} + C_d r^2 \label{eq:subharm} \\
    \gamma_r \ast g_{\Coulomb} &\leq C_d r^{2-d} \label{eq:growth} \\
    \langle \phi, g_{\Coulomb}\ast \phi\rangle &\geq 0\qquad \text{for any $\phi \in H^{-1}(\T^d)$} \label{eq:positivity}.
\end{align}
The first bound follows from the fact that $\Delta g \leq 1$, the second is an explicit calculation using the asymptotic \eqref{eq:GreenFuncAsymptotic_Torus} (and is only valid for $d\geq 3$),
and the third is equivalent to the positivity of the Fourier coefficients of $g$.  

For any configuration $X=(x_1,\cdots,x_N)$ we have
\begin{align*}
0 &\leq  \langle \gamma_r \ast \mu_X, g \ast (\gamma_r * \mu_X)\rangle \\
&= \langle \mu_X, (g\ast \gamma_r\ast \gamma_r) \ast \mu_X\rangle \\
&= \sum_{i\not=j} (g\ast \gamma_r \ast \gamma_r)(x_i,x_j) + N (g\ast \gamma_r\ast \gamma_r)(0,0)\\
&\leq \Big(\sum_{i\not= j} g(x_i,x_j)\Big) + C_d (N^2 r^2 + Nr^{2-d}).
\end{align*}
Taking $r=N^{-1/d}$ and observing that the sum in parentheses is proportional to $\mc{H}_N(X)$,
we conclude that for any $X$ we have
\[
-C_d N^{2-2/d} \leq \mc{H}_N(X),
\]
as desired. 

\section{Euclidean and external potential}
\label{sec:Rd}
We prove a similar bound on macroscopic fluctuations in the case of particles confined by an external potential. See \Cref{thm:fluctuations_Rd} for the result. 

\subsection{Setup}
We work in $d\geq 3$, and take $g(x) = |x|^{2-d}$ to be the Coulomb potential.


Let $V: \R^d \to \R$ be an external potential. For $X = (x_1, \ldots, x_N)$ a configuration of $N$ points, the total energy is 
\begin{align*}
    \mc H_N(X) := \frac{1}{2}\sum_{j\neq k}g(x_j - x_k) + N\sum_{1\leq j \leq N} V(x_j). 
\end{align*}
We consider the Gibbs measure 
\[
d\P_{N,\beta}(X) = \frac{1}{Z_{N,\beta}} e^{-\beta \mc H_N(X)}\, dX
\]
where the partition function 
\[
Z_{N,\beta} = \int_{(\R^d)^N} e^{-\beta \mc H_N(X)}\, dX
\]
is bounded as long as $\int \exp(-V(x))\, dx < \infty$. 
The energy of a non-atomic measure $\mu$ with bounded variation is 
\begin{align*}
    \mc E(\mu) := \frac{1}{2}\int \int g(x-y)\, d\mu(x) d\mu(y) + \langle 1, \mu\rangle \int V(x)\, d\mu(x).
\end{align*}
We define the total potential field $P\mu$ generated by a measure $\mu$, 
\begin{align*}
    P\mu(x) := \int g(x-y) d\mu(y) + \langle V, \mu\rangle + \langle 1, \mu\rangle V(x).
\end{align*}
Note that $P$ is the self-adjoint operator (on $L^2(\R^d)$) satisfying
\begin{align*}
    \mc E(\mu) = \frac{1}{2}\langle P\mu, \mu\rangle.
\end{align*}

The following theorem describes measures minimizing the energy functional. It is due to Frostman, see \cite{Serfaty}*{Theorem 2.1}.
\begin{theorem}\label{thm:FrostmanThm}
Assume $V$ is lower semi-continuous and bounded below, $\lim_{|x|\to \infty} V(x) = \infty$,  and $\{x\in \R^d\, :\, V(x) < +\infty\}$ has positive capacity. Then the minimum of $\mc E(\mu)$ over Borel probability measures exists, is
finite, and is achieved by a unique $\mu_V$ which has compact support. Moreover, letting $\zeta = P\mu_V$, 
\begin{align*}
    \begin{cases}
        \zeta \geq 2\mc E(\mu_V) & \text{almost everywhere on $\R^d$} \\ 
        \zeta = 2\mc E(\mu_V) & \text{almost everywhere on $\supp \mu_V$}.
    \end{cases}
\end{align*}
\end{theorem}
For example, if $d = 3$ and $g(x) = |x|^{-1}$ is the Coulomb potential and $V(x) = \frac{1}{2}|x|^2$ is the quadratic potential, then $\mu_V$ is the Lebesgue probability measure on a ball centered at the origin. 
Following Serfaty~\cite{Serfaty}, we assume $V$ satisfies the following properties. The assumptions simplify compared to \cite{Serfaty} because we assumed $g \geq 0$. 
\begin{enumerate}[label=(B\arabic*)]
    \item $V$ is lower semi-continuous and bounded below, \label{ass:V_1}

    \item $\lim_{|x|\to \infty} V(x) = \infty$, \label{ass:V_2}

    \item $\{x\in \R^d\, :\, V(x) < +\infty\}$ has positive capacity, \label{ass:V_3} 

    \item $\int_{\R^d} \exp(-V(x))\, dx < \infty$, \label{ass:V_4}

    \item The equilibrium measure is bounded in $L^{\infty}$, that is, $\mu_V = \mu_V(x) dx$ for some $\mu_V(x) \in L^{\infty}(\R^d)$. \label{ass:V_5}

\end{enumerate}
Notice that \ref{ass:V_1}--\ref{ass:V_3} are the hypotheses of \Cref{thm:FrostmanThm}.
\noindent 
We change $V$ by a constant if necessary so that $\mc E(\mu_V) = 0$. We define 
\begin{align*}
    \Sigma := \supp \mu_V.
\end{align*}
Recall
\begin{align*}
    \zeta := P\mu_V = g * \mu_V + \langle V, \mu\rangle + V.
\end{align*}
Then
\begin{itemize}
    \item $\zeta = 0$ on $\Sigma$ due to our normalization $\mc E(\mu_V) = 0$.
    \item $g * \mu_V \to 0$ as $x \to \infty$ because $\mu_V$ has compact support and $g$ decays to zero.
    \item $g * \mu_V$ is bounded above due to the hypothesis that $\mu_V$ has bounded density.
    \item It follows from the above fact that 
    \begin{align}\label{eq:zeta_like_V}
        |\zeta(x) - V(x)| \leq C
    \end{align}
    for some constant $C$. 
\end{itemize}
We consider the following normalized ground state energy, 
\begin{align}\label{eq:L_N_defn}
    L_N := \min_X ( \mc H_N(X) - N \langle \zeta, \mu_X\rangle ),
\end{align} 
where we penalize points of $X$ lying outside of $\Sigma$ via the effective potential $\zeta$. Serfaty provides the following lower bound on $L_N$. 
\begin{proposition}[\cite{Serfaty}*{Corollary 5.5}]\label{prop:energy_lower_bd}
For the Coulomb interaction, the ground state energy $L_N$ satisfies
\begin{align*}
    L_N \geq - C  N^{2-2/d} \| \mu_V \|_{L^{\infty}}^{\frac{d-2}{d}} 
\end{align*}
\end{proposition}
This bound serves as the Euclidean analogue of our torus bound on the minimum energy (Lemma~\ref{lem:min-energy-torus}). The difference is the term $N \langle \zeta, \mu_X\rangle$ added to the 
energy.  Fortunately, Serfaty also provides a tail bound for this quantity.
\begin{proposition}[\cite{Serfaty}*{Corollary 5.26}]\label{prop:confinement_by_zeta}
Assume \ref{ass:V_1}-\ref{ass:V_4} so that $\mu_V$ exists and is compactly supported. Assume also that $\mu_V$ has a bounded density. Then for all $\beta > 0$ we have 
\begin{align*}
    \Bigl|\log \E_{P_{N,\beta}} \Bigl[\exp \frac{1}{2}\beta N \langle \zeta, \mu_X\rangle\Bigr]\Bigr| \leq C \beta  N^{2-2/d} + C_{\zeta} N 
\end{align*}
where $C > 0$ depends only on $d, \|\mu_V\|_{L^{\infty}}$ and $C_{\zeta}$ depends on $\zeta$ and \ref{ass:V_4}. 
\end{proposition}

\subsection{Fluctuations of the potential}
Our main estimate is a bound on the exponential moment of the potential.
\begin{proposition}\label{prop:exponential_moment_estimate}
Assume \ref{ass:V_1}--\ref{ass:V_5}. If $\beta \geq 1/(N-1)$,
\begin{align*}
    \E_{X\sim \P_{N,\beta}}\Bigl[\Bigl(\int e^{-\beta P\mu_{X}(x) }\, dx\Bigr)e^{\beta\langle \zeta, \mu_X\rangle }\Bigr] \leq Ce^{-\frac{2\beta}{N} L_N+\beta C}.
\end{align*}
\end{proposition}
\begin{proof}
\noindent As in the torus, we integrate one $x$-variable at a time. Let
\begin{align*}
    X_{\hat j} &= (x_1, \ldots, x_{j-1}, x_{j+1}, \ldots, x_N) \\ 
    \mu_{X,\hat j} &= \sum_{k\neq j} \delta_{x_k}.
\end{align*}
The local energy experienced by $x_j$ is 
\begin{align*}
    P\mu_{X, \hat j}(x_j) &= \sum_{k\neq j} g(x - x_k)  + \sum_{k \neq j} V(x_k) + (N-1)V(x_j). 
\end{align*}
If we replace $x_j$ with a new point $x$, the global energy is given by
\begin{align}
    \mc H_N(x, X_{\hat j}) := \mc H_N(x_1, \ldots, x_{j-1}, x, x_{j+1}, \ldots, x_N) &= \sum_{k\neq j} g(x - x_k) + NV(x) + \text{(function of $X_{\hat j}$)} \nonumber \\ 
    &= P \mu_{X,\hat j} (x) + V(x) + \text{(function of $X_{\hat j}$)}.\label{eq:global_energy_local}
\end{align}
The global energy is also a sum of local energies, 
\begin{align}\label{eq:energy_identity}
    \mc H_N(X) = \frac{1}{2} \sum_{1\leq j \leq N} ( P \mu_{X,\hat j}(x_j) + 2V(x_j)).
\end{align}
We define the conditional Gibbs measure 
\[
d\P_{N,\beta | X_{\hat j}}(x) = \frac{e^{-\mc H_N(x, X_{\hat j})}}{\int_{\R^d} e^{-\mc H_N(x, X_{\hat j})}\, dx }.
\]
Using equation \eqref{eq:global_energy_local} for the Hamiltonian, we compute the exponential moment
\begin{align*}
    \E_{x\sim \P_{X_{\hat j}, N,\beta}} [e^{\beta P (\mu_{X,\hat j}(x) + V(x)) - \zeta(x)}] &= \frac{\int e^{-\zeta(x)} e^{\beta (P\mu_{X,\hat j}(x)+V(x)) }e^{-\beta \mc H_N(x, X_{\hat j})}\, dx}{\int e^{-\beta \mc H_N(x, X_{\hat j})}\, dx} \\ 
    &= \frac{\int e^{-\zeta(x)}\, dx}{\int e^{-\beta (P\mu_{X,\hat j} + V(x))}\, dx}.
\end{align*}
By \ref{ass:V_4} and \eqref{eq:zeta_like_V} the numerator is bounded, so integrating over the marginal distribution on $X_{\hat j}$ we find
\begin{align*}
    \E_{X\sim \P_{N,\beta}}\Bigl[\Bigl(\int e^{-\beta P\mu_{X,\hat j}(x) - \beta V(x)}\, dx\Bigr)\E_{x\sim \P_{X_{\hat j}, N,\beta}} [e^{\beta P (\mu_{X,\hat j}(x) + V(x)) - \zeta(x)}]\Bigr] \leq C. 
\end{align*}
We may remove the expected value on the inside to obtain 
\begin{align*}
    \E_{X\sim \P_{N,\beta}}\Bigl[\Bigl(\int e^{-\beta P\mu_{X,\hat j}(x) - \beta V(x)}\, dx\Bigr)e^{\beta P (\mu_{X,\hat j}(x_j) + V(x_j)) - \zeta(x_j)}\Bigr] \leq C. 
\end{align*}
For any $X = (x_1, \ldots, x_j, \ldots, x_N)$ and $x \neq x_j$,
\begin{align*}
    P\mu_{X,\hat j}(x) + V(x) &= P\mu_X(x) - g(x - x_j) - V(x) - V(x_j) \\ 
    &\leq P\mu_X(x) + C
\end{align*}
because $g \geq 0$ and $V$ is bounded below. Thus
\begin{align*}
    \int e^{-\beta P\mu_{X,\hat j}(x) - \beta V(x)}\, dx &\geq e^{-C \beta} \int e^{-\beta P\mu_X(x)}\, dx. 
\end{align*}
Using this pointwise bound in our exponential moment calculation gives
\begin{align*}
    \E_{X\sim \P_{N,\beta}}[\Bigl(\int e^{-\beta P\mu_X(x)}\, dx\Bigr) e^{\beta (P (\mu_{X,\hat j}(x_j) + V(x_j)) - \zeta(x_j)} ] \leq Ce^{C \beta}.
\end{align*}
The left hand side is the same for all $1 \leq j \leq N$. Summing and using Jensen's we find 
\begin{align*}
    Ce^{C \beta} &\geq  \E_{X\sim \P_{N,\beta}}\Bigl[\Bigl(\int e^{-\beta P\mu_X(x)}\, dx\Bigr)\frac{1}{N}\sum_{j=1}^N e^{\beta (P \mu_{X,\hat j}(x_j) + V(x_j)) - \zeta(x_j)} \Bigr] \\ 
    &\geq  \E_{X\sim \P_{N,\beta}}\Bigl[\Bigl(\int e^{-\beta P\mu_X(x)}\, dx\Bigr)\exp\Bigl(\frac{1}{N}\sum_{j=1}^N  (\beta P \mu_{X,\hat j}(x_j) + \beta V(x_j) - \zeta(x_j))\Bigr) \Bigr]. 
\end{align*}
We expand the exponential term as
\small
\begin{align*}
    \frac{1}{N}\sum_{j=1}^N  (\beta P (\mu_{X,\hat j}(x_j) + \beta V(x_j) - \zeta(x_j)) &= \frac{2 \beta}{N}\mc H_N(X) - \frac{\beta}{N} \langle V, \mu_X \rangle  - \frac{1}{N}\langle \zeta, \mu_X\rangle && \text{by \eqref{eq:energy_identity},}\\
    &\geq \frac{2\beta}{N}\mc H_N(X) - \frac{\beta+1}{N}\langle \zeta, \mu_X\rangle - C \beta  && \text{by \eqref{eq:zeta_like_V},}\\ 
    &\geq (2\beta-(1+\beta)/N) \langle \zeta, \mu_X\rangle - C\beta + \frac{2\beta}{N} L_N .
\end{align*}
\normalsize
In the last line, we use the regularized minimum energy quantity $L_N$ defined in \eqref{eq:L_N_defn}.
Assuming $\beta \geq 1/(N-1)$,
\[
(2\beta-(1+\beta)/N) \langle \zeta, \mu_X\rangle \geq \beta \langle \zeta, \mu_X\rangle
\]
Rearranging yields \Cref{prop:exponential_moment_estimate}.
\end{proof}

\subsection{The $L^1$ bound}
\Cref{prop:exponential_moment_estimate} involves the integral $\int e^{-\beta P\mu_X(x)}\, dx$.  To obtain a bound for 
the fluctuation of linear statistics we want instead an exponential moment estimate for $\|P\mu_X\|_{L^1}$.  The following
lemma relates these two quantities deterministically.
\begin{lemma}\label{lem:L1_lower_bds_cond_partition}
\begin{align*}
    \int e^{-\beta P\mu_{X}(x)}\, dx \geq \frac{1}{C} e^{\frac{1}{2}\beta( \| P\mu_X\|_{L^1(\mu_V)} - \langle \mu_X, \zeta\rangle)} - \frac{1}{C}.
\end{align*}
\end{lemma}
\begin{proof}
We assumed $\mu_V \in L^{\infty}$, so $dx\geq \frac{1}{C} d\mu_V(x)$. Thus
\begin{align*}
    \int e^{-\beta P\mu_{X}(x)}\, dx &\geq \frac{1}{C}\int e^{-\beta P\mu_{X}(x)}\, d\mu_V(x) \\  
    &\geq \frac{1}{C} \int e^{-\beta (P\mu_{X})_-(x)}\, d\mu_V(x) - \frac{1}{C} && \text{where } (P\mu_{X})_-(x) = P\mu_{X}(x) 1_{P\mu_{X}(x) < 0}\\ 
    &\geq \frac{1}{C} \exp(-\beta \int (P\mu_{X})_-(x)\, d\mu_V(x)) - \frac{1}{C} && \text{by Jensen's inequality.}
\end{align*}
We estimate the $L^1$ norm using the integral of the negative part,
\begin{align*}
    \| P\mu_{X}\|_{L^1(\mu_V)} &= \langle P\mu_{X}, \mu_V\rangle - 2\int (P\mu_{X})_-(x)\, d\mu_V(x) \\ 
    &= \langle \mu_X, \zeta\rangle - 2\int (P\mu_{X})_-(x)\, d\mu_V(x) 
\end{align*}
where we used self-adjointness of $P$ in the second line.
Thus
\begin{align*}
    \int e^{-\beta P\mu_X(x)}\, dx\geq \frac{1}{C} e^{\frac{1}{2}\beta( \| P\mu_X\|_{L^1(\mu_V)} - \langle\zeta, \mu_X\rangle)} - \frac{1}{C}.
\end{align*}
\end{proof}
Combining the probabilistic estimate \Cref{prop:exponential_moment_estimate} and the $L^1$ estimate \Cref{lem:L1_lower_bds_cond_partition},
\begin{align*}
    Ce^{-\frac{2\beta}{N}L_N+C \beta } &\geq \E_{X\sim \P_{N,\beta}}\Bigl[\Bigl(\int e^{-\beta P\mu_{X}(x) }\, dx\Bigr)e^{\beta\langle \zeta, \mu_X\rangle }\Bigr]  \\ 
    &\geq \frac{1}{C}\E_{X\sim \P_{N,\beta}}\Bigl[\Bigl(e^{\frac{1}{2}\beta( \| P\mu_X\|_{L^1(\mu_V)} - \langle \zeta, \mu_X\rangle)} - 1\Bigr)e^{\beta\langle \zeta, \mu_X\rangle }\Bigr]  \\ 
    &\geq \frac{1}{C}\E_{X\sim \P_{N,\beta}}\Bigl[e^{\frac{1}{2}\beta \| P\mu_X\|_{L^1(\mu_V)} } \Bigr] - \frac{1}{C}\E_{X\sim \P_{N,\beta}}\Bigl[e^{\beta\langle \zeta, \mu_X\rangle }\Bigr]. 
\end{align*}
Using \Cref{prop:confinement_by_zeta}, 
\begin{align*}
\E_{X\sim \P_{N,\beta}}\Bigl[e^{\beta\langle \zeta, \mu_X\rangle }\Bigr] \leq \E_{X\sim \P_{N,\beta}}\Bigl[e^{\frac{1}{2}\beta N\langle \zeta, \mu_X\rangle }\Bigr]^{2/N} \leq e^{C\beta N^{s/d}}.
\end{align*}
Using \Cref{prop:energy_lower_bd}, $L_N\geq -CN^{2-2/d}$, so we find 
\begin{proposition}\label{prop:L1_estimate_expoenential_moment}
If $N\geq 2$ and $\beta \geq 1/(N-1)$,
\begin{align*}
\E_{X\sim \P_{N,\beta}}\Bigl[e^{\frac{1}{2}\beta \| P\mu_X\|_{L^1(\mu_V)} } \Bigr] \leq Ce^{C \beta N^{(d-2)/d} }.
\end{align*}
\end{proposition}

\subsection{Application to fluctuations}

Let $\varphi \in C^2_c(\Sigma)$ be an observable supported in $\Sigma = \supp \mu_V$. We are interested in the fluctuations of the linear observable
\[
\langle \mu_X - N \mu_V, \varphi\rangle = \sum_{1 \leq j \leq N} \varphi(x_j) - N \int \varphi\, d\mu_V. 
\]
Assume we are in $d\geq 3$ with the Coulomb interaction $g(x) = \frac{1}{d-2}|x|^{-(d-2)}$. Then $g$ inverts the Laplacian, so 
\begin{align*}
    P (-\Delta) \varphi &= - \varphi - \langle V, \Delta \varphi\rangle 1_{\R^d} - \langle 1, \Delta \varphi\rangle V \\ 
    &= \varphi - \langle V, \Delta \varphi\rangle 1_{\R^d}.
\end{align*}
 Thus
\begin{align*}
    \langle \mu_X - N \mu_V, \varphi\rangle &= \langle \mu_X - N\mu_V, P(-\Delta)\varphi\rangle \\ 
    &= \langle P(\mu_X - N\mu_V), (-\Delta) \varphi\rangle \\ 
    &= \langle P\mu_X  -N \zeta, (-\Delta) \varphi\rangle \\
    &= \langle P\mu_X, (-\Delta) \varphi\rangle && \text{Because $\supp \varphi \subset \Sigma$}.
\end{align*}
For the last term on the right, we have the estimate
\begin{align*}
    |\langle P\mu_X, (-\Delta) \varphi\rangle| 
    &\leq \| P\mu_X \|_{L^1(\mu_V)} \Bigl\|\frac{-\Delta \varphi(x)}{\mu_V(x)}\Bigr\|_{\infty} \\
    &\leq \| P\mu_X\|_{L^1(\mu_V)} \|\Delta \varphi\|_{L^\infty} \|\frac{1}{\mu_V}\|_{L^\infty(\Sigma)},
\end{align*}
where we write $d\mu_V(x) = \mu_V(x) dx$ and in the last line we used that $\varphi$ is supported in $\Sigma$.
Applying this into our fluctuation estimate, Proposition~\ref{prop:L1_estimate_expoenential_moment} completes the 
proof of Theorem~\ref{thm:fluctuations_Rd}.

Note that we used the fact that $\Delta$ is a local operator in the step $\langle \zeta, (-\Delta \varphi)\rangle = 0$. If $g$ were a Riesz rather than Coulomb potential, the analagous quantity would be $\langle \zeta, (-\Delta)^{s/2} \varphi\rangle$, and the same argument does not work because the fractional Laplacian is non-local.  This is the only step of the argument that breaks for the case of a Riesz potential, and it would be interesting to try and adapt this argument to the Riesz case.


\end{document}